\documentclass[11pt,a4paper]{article}

\usepackage{todonotes}
\usepackage{array}

\usepackage{amsmath,amssymb,enumerate,mathrsfs}

\usepackage{framed}
\usepackage{xspace}

\usepackage[utf8]{inputenc}

\usepackage[german,dutch,english]{babel}

\DeclareMathOperator\disc{disc}
\usepackage{amsthm}

\usepackage{tikz}

\usepackage{anysize}

\usepackage{subfigure}
\usepackage{import}

\newcommand {\Z}	  {\mathbb{Z}}
\newcommand {\Q}	  {\mathbb{Q}}
\newcommand {\R}	  {\mathbb{R}}
\newcommand {\N}	  {\mathbb{N}}

\newcommand {\ie}  {\textit{i.e.}}
\newcommand {\eg}  {\textit{e.g.}}
\newcommand {\etal} {\textit{et al.}}

\newcommand{\ocp}{\mathrm{ocp}}
\newcommand{\MVS}{{\tt MVS}\xspace}
\newcommand{\MVD}{{\tt MSD}\xspace}
\newcommand{\OCP}{{\tt OCP}\xspace}

\newcommand{\vol}{\operatorname{vol}}

\newcommand\A{\mathscr{A}}

\newcommand{\Dm}{\Delta_{\mathrm{max}}}
\newcommand{\Do}{\Delta_{\mathrm{out}}}

\newcommand{\E}{\mathscr{E}}
\newcommand{\K}{\mathcal{K}}
\newcommand{\B}{\mathscr{B}_d}

\renewcommand{\epsilon}{\varepsilon}

\renewcommand{\leq}{\leqslant}
\renewcommand{\geq}{\geqslant}

\theoremstyle{plain}
\newtheorem{theorem}{Theorem}
\newtheorem{lemma}[theorem]{Lemma}
\newtheorem{corollary}[theorem]{Corollary}

\newtheorem{remark}{Remark}

\theoremstyle{definition}

\providecommand{\conv}{\mathrm{conv}}

\title{On largest volume simplices and  sub-determinants}

\author{Marco Di Summa\thanks{Dipartimento di Matematica, Universit\`a degli Studi di
      Padova (Italy). \texttt{disumma@math.unipd.it}}
  \and
  Friedrich Eisenbrand\thanks{Ecole Polytechnique F\'ed\'erale de Lausanne (Switzerland). \texttt{friedrich.eisenbrand@epfl.ch, yuri.faenza@epfl.ch, carsten.moldenhauer@epfl.ch}}
  \and
  Yuri Faenza\footnotemark[2]
  \and
  Carsten Moldenhauer\footnotemark[2]
}

\date{\today }

\begin{document}
\maketitle

\begin{abstract}
 \noindent
 We show that the problem of finding the simplex of largest volume in the convex hull of $n$ points in $\Q^d$  can be approximated with a
 factor of $O(\log d)^{d/2}$ in polynomial time. This improves upon
 the previously best known approximation guarantee of $d^{(d-1)/2}$ by
 Khachiyan.

 On the other hand, we show that there exists a constant $c>1$ such
 that this problem cannot be approximated with a factor of $c^d$, unless $P=NP$.
 Our hardness result holds even if $n = O(d)$, in which case there exists a
 $\bar c\,^{d}$-approximation algorithm that relies on recent sampling
 techniques, where $\bar c$ is again a constant.

We show that similar results hold for the problem of finding the
 largest absolute value of a subdeterminant of a $d\times n$ matrix.
\end{abstract}


\section{Introduction}

Many techniques in convex geometry begin with approximating a
geometric shape by a simpler one. The maximum volume ellipsoid, or
\emph{John ellipsoid} (see, \eg,~\cite{GroetschelLovaszSchrijver88,matousek02}),
for example, is a prominent such simplification with  many applications in discrete and continuous
optimization. 

\medskip
Simplices are, next to ellipsoids, among the most primitive convex
sets.  We are interested here in  the problem of approximating a given
$V$-polytope by a contained simplex of largest volume.  More
precisely, we investigate the approximability and hardness of the
following problem.

\begin{framed}
Maximum Volume Simplex (\MVS)
\begin{quote}
  Given $n$ points $a_1, \dots, a_n \in \Q^d$, find a simplex
  of maximum volume that is contained in the convex hull $\conv\{a_1,\dots,a_n\}$  of these points.
\end{quote}
\end{framed}
%
\noindent
We assume here, without   loss of generality, that the convex hull of the
points $a_1,\dots,a_n$ is full-dimensional.
As is the case for ellipsoids, the largest volume simplex  in a convex body
has  attracted a lot of attention in the computer science and optimization
literature, see, \eg,~\cite{khachiyan95,gritzmann94,gritzmann95,brieden00,packer02,packer04}.

\medskip
 The volume $\vol(\Sigma)$ of a full-dimensional simplex
$\Sigma =\conv\{v_0,v_1,v_2,\dots,v_d \} \subseteq \R^d$ is
\begin{displaymath}
\vol(\Sigma)=  \frac{| \det(A) |}{d!},
\end{displaymath}
where $A$ is the matrix with columns $v_1 - v_0, v_2 - v_0, \dots, v_d
- v_0$.
Thus \MVS can be reduced to $n$ instances of the problem of finding the largest absolute value
of a $d \times d$ subdeterminant of a  $d\times (n-1)$ matrix. This
motivates the second problem that is central to our study.

\begin{framed}
Maximum  Subdeterminant (\MVD)
\begin{quote}
  Given a matrix $A \in \Q^{d \times n}$ of full row-rank, determine a basis $B \subseteq \{1,\dots,n\}$ of $A$ for which  $|\det(A_B)|$ is maximum.
\end{quote}
\end{framed}
\noindent
Here a \emph{basis} of a $d\times n$ matrix $A$ is a maximal
subset of the column indices $B \subseteq \{1,\dots,n\}$ such that the
corresponding columns are linearly independent, and $A_B$ is the matrix
consisting of the columns indexed by $B$. Khachiyan~\cite{khachiyan95}
has shown that there exists a $((1 + \epsilon)\cdot d)^{(d-1)/2}$
approximation algorithm for \MVD, and thus also for \MVS, with
running time polynomial in $n,d$ and $1 / \epsilon$.


\medskip
\noindent
Our  main contributions are as
follows.
\begin{enumerate}[(i)]
\item We show that there exists an algorithm for \MVS, with running time polynomial in $n,d$ and $1/
  \epsilon$, that computes a simplex $\Sigma$ with
  \begin{displaymath}
    \vol(\Sigma) \cdot \Big( e \ln\big( ( 1+\epsilon)  d\big) \Big)^{d/2} \geq \mathrm{Opt},
  \end{displaymath}
  where $\mathrm{Opt}$ denotes the maximum volume of a simplex
  contained in $\conv\{a_1,\dots,a_n\}$. To achieve this, we
  significantly tighten the analysis of Khachiyan's
  algorithm~\cite{khachiyan95}.
	We also show that our
  analysis is essentially tight by describing instances where the
  approximation ratio of Khachiyan's algorithm is $(\alpha \ln(d))^{d/2}$ with $\alpha \ge
  0.748$. \label{item:6}
\item We show that there exists a constant $c>1$ such that \MVS cannot be
  approximated within a factor of $c^d$, unless $P = NP$. This improves
   the previous best $1.09$-inapproximability of
  Packer~\cite{packer04}. \label{item:7}
\item Our hardness result~\eqref{item:7} holds for instances with $n =
  \Theta(d)$. A recent sampling technique~\cite{deshpande10} immediately
  yields a $\bar c\,^d$-approximation for such instances (for another constant $\bar c$),
  showing that the  hardness is essentially tight in this case.  \label{item:8}
\item These  results~\eqref{item:6}, \eqref{item:7} and~\eqref{item:8} also hold for \MVD. \label{item:9}
\end{enumerate}
%
%
%


\subsection{Related work}
\label{sec:related-work}

The literature on topics related to \MVS and \MVD is extensive. In order to put our results in perspective, we
provide an overview of a selection of related papers.

\subsubsection*{Approximating convex bodies}
Brieden, Gritzmann and Klee~\cite{brieden00} have shown that one can
compute a simplex $\Sigma$ in a convex body with $\vol(\Sigma) (d+1)^d
\geq \mathrm{Opt}$ if the convex body is equipped with a weak
separation oracle.  This is similar to the problem of computing a
maximum volume ellipsoid.  Computing the John-ellipsoid is in general
NP-hard, even if $K$ is a $V$-polytope, \ie, a polytope represented
by its vertices. However, one can compute an approximation of the
John-ellipsoid in polynomial time. Gr\"otschel, Lov\'asz and
Schrijver~\cite{GroetschelLovaszSchrijver88} have shown that, if a
convex set $K\subseteq \R^d$ is given by a weak separation oracle,
then one can compute in polynomial time an ellipsoid in $K$ that, if
scaled by a factor of roughly $d^{3/2}$, contains $K$. 
Thus there is an
approximation algorithm for the problem of computing a maximum volume
ellipsoid in convex sets with an approximation guarantee of
$d^{3d/2}$.
When $K$ is an $H$-polytope (\ie, a polytope described through a system of linear inequalities), then a nearly optimal algorithm is known
~\cite{khachiyan93}. 

Variants where the dimension of the solution can be restricted, \eg~finding a maximum volume $j$-dimensional
simplex, have been considered for \MVS and \MVD~\cite{gritzmann95,brieden00,packer04,civril09}.

\subsubsection*{Hardness of approximation}

Packer~\cite{packer04} has shown that \MVS is inapproximable within a
constant factor smaller than $1.09$.
This implies the same hardness for \MVD, which was shown earlier to be NP-hard by Papadimitriou~\cite{papadimitriou84}.

Koutis~\cite{koutis06} considered
the problem of finding the maximum volume $j$-dimensional simplex in a $V$-polytope.
Note that one obtains problem \MVS when $j=d$.
\c{C}ivril and Magdon-Ismail~\cite{civril13} considered the following
problem that is related to \MVD. Given a matrix $A \in \Q^{d\times n}$
and an integer~$j$, select a subset $J \subseteq \{1,\dots,n\}$ of cardinality~$j$
such that $\sqrt{\det(A_J^TA_J)}$ is maximized.
Note that one obtains problem \MVD when $j=d$.

In both cases, the authors show
that there exist a constant $c > 1$ and a function $j(d)$ such that it is
NP-hard to approximate the respective problem with factor less than
$c^{j(d)}$. Here, $j(d)$ is linear in $d$, but with constant
dependence strictly less than one, thus these results do not cover the case $j=d$.


\subsubsection*{Subdeterminants in optimization and combinatorics}
\label{sec:extr-determ-sub}

An integer matrix $A \in \Z^{d\times n}$ is \emph{totally unimodular}
if the largest absolute value $\Delta_k$ of a $k\times k$
subdeterminant of $A$ is at most one for each $k \in \{1,\dots,d\}$. This
is the case if and only if the optimum value of \MVD is one for the matrix $(A
\mid I_d)$, where $I_d$ is the identity matrix of size $d\times d$. Seymour~\cite{MR579077} provided a polynomial-time
algorithm that tests whether a matrix is totally
unimodular. \emph{Integer programs} $\max\{c^Tx \colon Ax \leq b, \,
x\geq 0,\, x \in \Z^n\}$ defined by totally unimodular matrices $A \in
\Z^{d\times n}$ can be solved in polynomial time.  Constraint matrices
with small subdeterminant also play an important role in convex
(integer) optimization~\cite{hochbaum1990convex}.

Subdeterminants are also fundamental in \emph{discrepancy
  theory}. The \emph{discrepancy} of a matrix $A \in
\R^{n\times d}$ is defined as $\disc(A) = \max_{x \in \{- 1,
  1\}^d}\|Ax \|_\infty$, see,
\eg,~\cite{matousek02,chazelle2000discrepancy}. The \emph{hereditary
  discrepancy} of $A$ is $\max_{S \subseteq [d]} \disc(A_S)$.  Very
recently there have been several breakthroughs in the field of
approximation algorithms related to
discrepancy. Bansal~\cite{bansal2010constructive} has shown how to
find a coloring that respects Spencer's bound~\cite{spencer1985six}.
The concept of hereditary discrepancy is closely related to \emph{LP rounding}~\cite{lovasz1986discrepancy}
and important in the area of approximation algorithms.
Rothvo\ss~\cite{rothvoss2013approximating} recently improved the long-standing
$O((\log n)^2)$ additive error of Karmarkar and
Karp~\cite{KarmakarKarp82} using techniques from discrepancy theory.

The \emph{subdeterminant lower bound for hereditary discrepancy is}
$\max_k \sqrt[k]{\Delta_k}$. Recently
Matousek~\cite{matouvsek2013determinant} has shown that the
subdeterminant bound is tight up to a polynomial factor in $\log d $
and $\log n$. In a recent series of papers by Nikolov \etal~\cite{nikolov2013geometry,nikolov2013approximating} it was shown
how to approximate the hereditary discrepancy, and thus the subdeterminant bound $\max_k
\sqrt[k]{\Delta_k}$, within a polynomial factor in $\log d $ and $\log
n$. Our result provides
a $O(\log d)$-approximation to $\sqrt[d]{\Delta_d}$. It is  an
interesting problem whether a polynomial time approximation algorithm
for the subdeterminant bound with a guarantee that is polynomial in
$\log d$ exists.


\section{A tight analysis of Khachiyan's algorithm}
\label{sec:appr-algor-mvd}

We now  come to the main algorithmic result of our paper and  show the following theorem. Recall that $A_B$ is the matrix corresponding to the columns of $A$ indexed by $B$. We denote the set of bases of $A$ by  $\mathscr{B}$.
\begin{theorem}
  \label{thr:2}
  There exists an algorithm that given a matrix $A \in \Q^{d\times n}$ and $\epsilon > 0$, identifies a basis $B\subseteq\{1,\dots,n\}$ such that

  \begin{displaymath}
|\det(A_B)| \cdot  \Big( e \ln\big( ( 1+\epsilon)  d\big) \Big)^{d/2} \geq \max_{B' \in \mathscr{B}} |\det(A_{B'})|.
  \end{displaymath}
 The algorithm runs in time polynomial in $n, d$, and $1/\epsilon$.
Thus \MVD and \MVS can be approximated within the factor above.
\end{theorem}
%
%
%
%
We first review
Khachiyan's algorithm~\cite{khachiyan95} for \MVS and his analysis. 
Suppose we are given a matrix $A \in \Q^{d\times
  n}$ of full row rank whose columns are $a_1,\dots,a_n$ respectively. Consider the symmetric polytope
\begin{displaymath}
  \A = \conv\{ \pm a_1,\dots,\pm a_n\}.
\end{displaymath}
The largest $d\times d$ subdeterminant of $A$, in absolute value,
corresponds to the largest volume simplex in $\A$ with one vertex being
the origin. The algorithm begins by \emph{rounding} the polytope~$\A$.


Khachiyan~\cite{khachiyan96} showed that, if $\mathcal K$ is a
symmetric  $V$-polytope that is explicitly given by its vertices, then it is
possible to compute an ellipsoid $\E$ such that $\E\subseteq \K
\subseteq \sqrt{(1+\epsilon)d}\ \E$ in time polynomial in the number
of vertices of $\mathcal K$, $d$ and $1 / \epsilon$. This applies to the polytope $\A$.

The rounding step is now as follows. Compute an approximation of the
ellipsoid $\E$ with  $\E\subseteq \A
\subseteq \sqrt{(1+\epsilon)d}\ \E$.
 Now, there exists a non-singular
matrix $T \in \R^{d \times d}$ such that the image of $\E$ is the
$d$-dimensional unit ball $\B$. The rounded instance of \MVD is then $T
\cdot A$. Clearly, this transformation is approximation preserving, since $\det((T\cdot A)_B) = \det(T) \cdot \det(A_B)$
for every basis $B$ of $A$.
Assume we have performed this rounding step. Then
\begin{displaymath}
  \B \subseteq \A \subseteq  \sqrt{(1+\epsilon) d}\, \B
\end{displaymath}
holds. 
From there, the algorithm  proceeds in a greedy fashion, see Figure~\ref{fig:kha}.

\begin{figure}[h]
	\smallskip\hrule\smallskip
	Input: Matrix $A\in\Q^{d\times n}$ with columns $a_1,\dots,a_n$ and $\epsilon>0$.
	    \smallskip\hrule\smallskip
	\begin{enumerate}[1.]
		\item Round the instance such that $\B \subseteq\mathcal{A}\subseteq\sqrt{(1+\epsilon)d}\,\B$,
			where $\A = \conv\{ \pm a_1,\dots,\pm a_n\}$ and $\B$ is the $d$-dimensional unit ball.
        \item For $i=1,\dots,d$:
                        	\begin{enumerate}[2.1.]
			\item Pick $v_i$ as the vector from $\{a_1,\dots,a_n\}$ with largest norm;
			\item Replace vectors $a_1,\dots,  a_n$ with their projections onto the orthogonal
                          complement of $v_i$.
		\end{enumerate}
		\item Return the original vectors corresponding to $\{ v_1,\dots, v_d\}$.
	\end{enumerate}
\hrule
\caption{Khachiyan's algorithm for \MVD.}
\label{fig:kha}
\end{figure}
%

Note that
$v_1,\dots,v_d$ correspond to the Gram-Schmidt orthogonalization of
the vectors returned by the greedy procedure. Thus, after rounding, the
absolute value of the determinant induced by the chosen vectors is
$\|v_1\|\dots\|v_d\|$.  

\subsection{Khachiyan's analysis}
\label{sec:khachiyans-analysis}

We now review Khachiyan's analysis showing that his algorithm is a
factor $((1+\epsilon)d)^{(d-1)/2}$ approximation algorithm for \MVD.
Let us denote the Euclidean lengths of the picked vectors $v_i$ by
$\rho_i = \lVert v_i \rVert$ for $i=1,\dots,d$, and define
$\Delta_{\max}=\max_{B \in \mathscr{B}}|\det(A_B)|$.  The claimed
bound follows from the following two facts:
\begin{enumerate}[(i)]
\item $\rho_1^d \geq \Dm$,  \label{item:1}
\item $\rho_i \geq 1$ for $i=1,\dots,d$.  \label{item:2}
\end{enumerate}
Fact~\eqref{item:1} is the well known \emph{Hadamard bound}, see,
\eg,~\cite{Schrijver98}, while \eqref{item:2} follows from the fact
that the (lower dimensional) unit ball continues to be included in the
convex hull of the projection of the input vectors to a
lower-dimensional space.  The output of the algorithm is the
subdeterminant $\Do = \rho_1\cdots \rho_d$.  By
combining~\eqref{item:1} and~\eqref{item:2} with $\rho_1 \leq
\sqrt{(1+\epsilon)d}$ one obtains $\Do \cdot
\left( \sqrt{(1+\epsilon)d} \right)^{d-1} \geq \Dm$, which is the
claimed approximation ratio.

\begin{theorem}[Khachiyan~\cite{khachiyan95}]
  \label{thr:1}
  There is a polynomial-time $((1+\epsilon)d)^{(d-1)/2}$-approximation algorithm for problems \MVD and \MVS.
\end{theorem}


\subsection{Improving the analysis}
\label{sec:an-improved-analysis}

We will now show that the approximation factor 
can in fact  be bounded by $(e \ln( ( 1+\epsilon) d))^{d/2}$. To
do so, we significantly tighten the upper bound~\eqref{item:1}
presented in the analysis above. Key to this improvement is the following observation.

\begin{lemma}
  \label{lem:1}
  Let $a_1,\dots,a_n \in \R^d$ be the input vectors  and let
  $v_1,\dots,v_d$ be the picked vectors in the course of Khachiyan's
  algorithm, with lengths $\rho_1,\dots,\rho_d$ respectively. Let $\E$ be an ellipsoid with the following properties:
  \begin{enumerate}[\upshape(a)]
  \item Each $v_i$ is on a principal axis of $\E$.  \label{item:5}
  \item None of the vectors $v_i$ is contained in the interior of  $\E$. \label{item:3}
  \item There exists $\alpha >0$ such that each $a_i$ is contained in $\alpha \cdot \E$. \label{item:4}
  \end{enumerate}
Then $\Dm \leq \alpha^d \rho_1 \cdots \rho_d$.
\end{lemma}

\begin{proof}
  We first observe that the largest $d\times d$ subdeterminant of a
  matrix whose columns are in $\E$ is bounded by $\rho_1 \cdots
  \rho_d$. To see this, let $p_i$ be an intersection point of the
  principal axis that includes $v_i$, with the boundary of $\E$. Since
  $v_i$ is not in the interior of $\E$ we have $\|p_i\|\le \rho_i$.
  Let $T$ be the inverse of the matrix with columns $p_1,\dots
  ,p_d$. In fact, $T$ has rows $p_1^T / \|p_1\|^2, \dots, p_d^T/
  \|p_d\|^2$.  The transformation $x \mapsto T\,x$ maps $\E$ to the
  unit ball $\B$. By the Hadamard bound, a selection of $d$ vectors in
  $\B$ has a determinant of at most $1$ in absolute value.  This
  implies that the largest $d\times d$ subdeterminant of a matrix
  whose columns are in $\E$ is bounded by $\|p_1\|\cdots\|p_d\|\le
  \rho_1 \cdots \rho_d$.  Now, since $\alpha \cdot \E$ contains all
  the input vectors, we have $\Dm \leq \alpha^d \rho_1 \cdots \rho_d$.
\end{proof}
\noindent
We can now prove our main result.

\begin{proof}[Proof of  Theorem~\ref{thr:2}]
The theorem follows from the existence of an ellipsoid
   $\E$ satisfying the conditions of  Lemma~\ref{lem:1}  with $\alpha =\sqrt{e
  \ln((1+\epsilon)d)}$.

The vectors $v_1,\dots,v_d$ are pairwise orthogonal due to the projection in Step 2.
Recall that determinants are invariant under rotation and note that the algorithm's execution
is also invariant under rotation. Therefore, by suitably rotating the input,
we can assume without loss of generality
that $v_1,\dots,v_d$ are the vectors $\rho_1 \cdot e_1, \dots, \rho_d \cdot e_d$,
where $e_i$ is the $i$-th unit vector.

Next, let us \emph{group the elements} of the decreasing sequence $\rho_1,\dots,\rho_n$
as follows. Let $G_j$ be the set of indices $i$ such that
$\frac{\rho_1}{\exp((j-1)/2)} \ge \rho_i >
\frac{\rho_1}{\exp(j/2)}$. Let $t$ denote the number of such groups.
Since $\rho_1/\rho_d \le  \sqrt{(1+\epsilon)d}$, we have $t\le
\log_{\exp(1/2)} \sqrt{(1+\epsilon)d} = \ln ((1+\epsilon)d)$.  Assume
that all groups $G_j$ are non-empty (discarding empty groups will
decrease the number of groups and subsequently yield an improved
analysis).  Let $r_j$ be the largest element of $G_j$ and note that
$r_j/\rho_i \le \sqrt e$ for all $\rho_i\in G_j$.

Let us decompose every $x\in\R^d$
into $x =  x^{(1)} + \dots + x^{(t)}$, where $ x^{(1)}\in\R^d$ is equal to $x$ in the first $|G_1|$ components
and zero otherwise, $x^{(2)}$ is equal to $x$ in the next $|G_2|$ components and zero otherwise, and so forth.
Let $\E$ be the ellipsoid
\begin{displaymath}
  \E = \left\{ x \in\R^d\, : \left\|\frac{
        x^{(1)}}{r_1 / \sqrt{e} }\right\|^2+\dots+\left\|\frac{
        x^{(t)}}{r_t / \sqrt{e}}\right\|^2 \le 1 \right\}.
\end{displaymath}
We are done once we have shown that
  the ellipsoid $\E$, together with  $\alpha =\sqrt{e
  \ln((1+\epsilon)d)}$,  satisfies  the conditions \eqref{item:5}--\eqref{item:4} from Lemma~\ref{lem:1}.

	First, the principal axes of $\E$ are the coordinate directions $\{ \lambda \cdot e_i \colon \lambda \in \R\}$,
  which implies \eqref{item:5}.
	Second, suppose that $i \in G_j$ and recall that $v_i = \rho_i \cdot e_i$. Since
  $\rho_i \ge r_j / \sqrt{e} $, it follows that $v_i$ is not in the interior of
  $\E$, which implies \eqref{item:3}. Every input column $a$ satisfies
  $\| a^{(j)} /r_j \| \leq 1$, otherwise the projection of $a$
  would have been picked instead of the vector corresponding to the
  first element of $G_j$. Consequently, each input column $a$ satisfies the constraint
  \begin{displaymath}
    \left\|\frac{a^{(1)}}{r_1 / \sqrt{e}}\right\|^2+\dots+\left\|\frac{a^{(t)}}{r_t / \sqrt{e}}\right\|^2 \le e \cdot t \le e
  \ln((1+\epsilon)d),
  \end{displaymath}
  which implies \eqref{item:4}.

\end{proof}


\subsection{Tightness of the analysis}

We now provide a family of instances of increasing dimension where
Khachiyan's algorithm achieves a ratio of $(\alpha \ln(d))^{d/2}$,
where $\alpha \ge 0.748$ tends to $1$ as $d$ grows. Thus, we basically
match the upper bound on the approximation ratio given in the previous
section.

Let $d\geq 4$ be a power
of two. The instances are matrices with $d$ rows of the form
\begin{equation}
  \label{eq:1}
  A = \left[\begin{matrix} D \quad|& \frac{\sqrt{c^2-1}}{c} D H \quad|& E \end{matrix}\right].
\end{equation}
Here, $D$ is a $d\times d$ diagonal matrix of the form
\begin{displaymath}
  D =
  \begin{pmatrix}
    c^{d-1} &  &  &  & \\
           & c^{d-2}&  &&\\
           & & \ddots & & \\
           & &        & c^0
  \end{pmatrix}
\end{displaymath}
with $c = {d}^\frac{1}{2 (d-1)}$.  Each column of $D$ has Euclidean
norm of at most $\sqrt{d}$.  The matrix $H$ is a \emph{Hadamard}
matrix, \ie, $H\in\{-1,1\}^{d\times d}$ with $| \det (H) | =
d^{d/2}$ (it is well-known that such a matrix certainly exists if $d$ is a power of two).
Also the Euclidean norm of each column of
$\frac{\sqrt{c^2-1}}{c} D H$ is bounded by $\sqrt{d}$. The matrix $E$
is such that the norm of each column is at most $c$ and
the unit ball is contained in $\conv(E)$. Clearly, such a matrix $E$ exists, as $c>1$; see, \eg,~\cite{dudley74} for explicit bounds.

Thus the polytope that is generated by the columns of the
matrix~\eqref{eq:1} and their negatives is ``round'',
in the sense that it contains the unit
ball and it is contained in the unit ball scaled by
$\sqrt{d}$.
We will show below that Khachiyan's algorithm will output a solution of value at most $c \cdot |\det D|$. This implies the claim, as
\begin{displaymath}
\left|\det\left(  \frac{\sqrt{c^2-1}}{c} D H\right)  / (c\det(D))\right| = \left(\frac{c^2-1}{c^2}\right)^{d/2} \frac{d^{d/2}}{c}.
\end{displaymath}
Then using $x-1\ge \ln(x)$, we deduce that the approximation ratio is
\begin{align*}
  d^{-\frac{1}{2(d-1)}} \left(\frac{d^{1/(d-1)}-1}{d^{1/(d-1)}}
    d\right)^{d/2} \ge \left( d^{\frac{d-2}{d-1}-\frac{1}{d(d-1)}}\
    \frac{1}{d-1} \ln(d) \right)^{d/2} \ge \left(\alpha
    \ln(d)\right)^{d/2},
\end{align*}
where $\alpha$ is as required.

Let $w$ be any column vector of $\frac{\sqrt{c^2-1}}{c} D H$. The squared norm of the projection of $w$
into the orthogonal complement of the first $i$ column vectors of $D$ (where $i\in\{0,\dots,d-1\}$) is
\begin{align*}
   \frac{c^2-1}{c^2} \sum_{j=0}^{d-i-1} c^{2j}  = \frac{c^2-1}{c^2}  \cdot \frac{c^{2(d-i)} -1}{c^2-1}
         = \frac{1}{c^2} \cdot (c^{2(d-i)}-1) < c^{2(d-i-1)}.
\end{align*}
The last term is the squared norm of the $(i+1)$-th column vector of $D$. Thus Khachiyan's algorithm outputs the first $d-1$ columns of $D$, and in the last step a column of norm at most $c$ from $E$. The determinant of the column vector selected by Khachiyan's algorithm is then at most $c\times|\det D|$, as claimed above.

\section{Hardness}\label{sec:hardness}

We now consider the hardness of approximating \MVS and \MVD.  As
mentioned in Section~\ref{sec:related-work}, the best
inapproximability result was due to Packer~\cite{packer04}, who proved
that \MVS cannot be approximated with a factor better than $1.09$,
unless $P=NP$.  In this section we provide a drastic improvement
showing that it is NP-hard to approximate \MVS and \MVD with a factor
$c^d$, where $c > 1$ is an explicit constant. In particular, we show
that the result holds for instances where $n=\Theta(d)$. We will also conclude that the hardness result is best possible for such instances.

\smallskip
Our argument is based on the connection between \MVD and the following problem.

\begin{framed}
Odd Cycle Packing (\OCP)
\begin{quote}
  Given a simple undirected graph, find a maximum family of vertex-disjoint odd cycles.
\end{quote}
\end{framed}

In fact, given a graph $G$, let $A_G$ be the node-edge incidence matrix of $G$.
Then, for every odd cycle $C$ of $G$, the square submatrix of $A_G$ with rows corresponding
to the nodes of $C$ and columns corresponding to the edges of $G$ has determinant $\pm2$.
Therefore any collection of $k$ vertex-disjoint odd cycles in $G$ determines a submatrix of $A_G$
whose determinant is $2^k$ in absolute value. This implies that $\Delta_{\max}(A_G)\ge2^{\ocp(G)}$,
where $\ocp(G)$ denotes the optimal value of \OCP on $G$.
Conversely, all non-zero subdeterminants of $A_G$ are
powers of two (in absolute value) and indeed one can show that $\Delta_{\max}(A_G)=2^{\ocp(G)}$ (see, \eg, \cite{grossman95}).

\smallskip
The overall strategy for proving hardness is the following. We first build on a hardness result by Berman and Karpinski~\cite{berman03} on stable sets in $3$-regular graphs and show that \OCP is NP-Hard to approximate with a factor ${\bar c}$, where $\bar c > 1$ is an explicit constant.
Our second step is to reduce \OCP to \MVD, using the construction seen above. Hence the constant inapproximability for \OCP leads to a $c^d$-inapproximability for \MVD. Last, we reduce \MVD to \MVS.

Let us remark that \OCP is NP-hard even when restricted to planar graphs~\cite{hardy05}
and, in that case, allows for constant factor approximations~\cite{fiorini07,kral12}.
Following the hardness for packing the maximum number of disjoint cycles by Friggstad and Salavatipour~\cite{salavatipour11},
we can deduce for \OCP a constant hardness under P$\neq$NP and a hardness of $O\left(\log^{\frac 1 2 - \epsilon} n\right)$
unless $\text{NP} \subseteq \text{ZPTIME}\left(n^{\text{polylog}(n)}\right)$, where $n$ is the number of nodes of the graph.
This result relies on the PCP-theorem. Our construction below yields a weaker hardness for \OCP,
but it does not rely on the PCP-theorem, it leads to an improved hardness result for the vertex-disjoint triangle packing problem, and it is much simpler to use for constructing a hardness for \MVD subsequently.
In particular, it enables us to easily calculate the explicit constants in the hardness for \MVD.

On the positive side, Kawarabayashi and Reed~\cite{KR10} showed that
for general graphs \OCP can be approximated within a factor of $\sqrt
n$.

\subsection{From stable set in $3$-regular graphs to $\OCP$}\label{sec:ss-to-ocp}

We now describe our inapproximability result for \OCP.
We require a result of Berman and Karpinski~\cite{berman03} for maximum stable set on 3-regular graphs.
Given a system of $2n$ linear equations modulo 2 with 3 variables per equation, H{\aa}stad~\cite{hastad97}
showed that it is NP-hard to distinguish between instances where there exists a solution satisfying $(1-\epsilon)2n$
equations and instances where no solution satisfies more than $(1+\epsilon)n$ equations, for any arbitrarily small $\epsilon>0$.
Building on this result, Berman and Karpinski~\cite{berman03} gave a polynomial time construction of a 3-regular graph
on $176n$ vertices that translates $n$ satisfied equations to a maximum stable set of size at most $97n$ and $2n$
satisfied equations to a maximum stable set of size at least $98n$.
Thus, it is NP-Hard to detect if a $3$-regular graph on $176n$ vertices has a stable set of size at least $(1-\epsilon)98n$, or at most $(1+\epsilon)97n$, for each $\epsilon >0$.

Let $G=(V,E)$ be the graph constructed in~\cite{berman03}. Intuitively, we would like to construct a new graph $H$ such
that every vertex in $G$ corresponds to a triangle in $H$ and that a stable set in $G$ also corresponds to a packing of triangles in $H$.
A first candidate for such a graph $H$ would be the line graph of $G$ (recall that $G$ is $3$-regular), but in the line graph
we might also create triangles that do not correspond to vertices in $G$.

We solve this issue by slightly changing $G$. Subdivide every edge in $G$ twice, \ie, substitute an edge $\{u,v\}$ by
a path $P_{uv} = u,p_1,p_2,v$. Let $G'$ be the obtained graph.
Since $G$ has $176n$ vertices, hence $\frac{3}{2}176n = 264n$ edges, $G'$ has $176n+2 \cdot 264n = 704n$ vertices
and $3 \cdot 264n = 792n$ edges. Note that $G'$ is triangle-free.

We now prove that every subdivision of an edge in $G$ augments the stable sets by exactly one vertex.
Consider a stable set $S$ in $G$. By choosing $p_2$ when $u\in S$ or $p_1$ otherwise we obtain an induced stable set
in $G'$ of size $|S|+264n$. Conversely, let $S'$ be a stable set in $G'$. Modify $S'$ such that for each $P_{uv}$ not
both $u$ and $v$ are in $S'$: if both are in $S'$, then the stable set $S'\setminus \{v\}\cup\{p_2\}$
has the same cardinality and only includes one of $\{u,v\}$.
Then we can obtain a stable set in $G$ of size at most $|S'|-264n$.
In particular, a stable set of size $97n$ (resp., $98n$) in $G$ translates into a stable set of size $97n+264n=361n$ (resp., $362n$) in $G'$.


Now let $H$ be constructed as follows: starting from the line graph of $G'$, for each $P_{uv}$ add two vertices and
connect them as to obtain the graph $H_{uv}$ depicted in Figure~\ref{fig:edge_dual}.
The number of vertices in $H$ is $792n + 2 \cdot 264n = 1320n$ and the number of edges is $\frac{1}{2}(4 \cdot 792n + 2 \cdot 2 \cdot 264n) = 2112n$,
as every vertex belonging to the line graph of $G'$ has degree four and the two additional vertices in each $H_{uv}$ have degree two.

As $G'$ is triangle-free, there is a one-to-one correspondence between
triangles in $H$ and vertices in $G'$. Moreover, two vertices in $G'$
are adjacent if and only if their corresponding triangles in $H$ have
a common vertex.  Thus a maximum stable set of size $361n$ (resp.,
$362n$) in $G'$ translates into a maximum number of $361n$ (resp.,
$362n$) vertex-disjoint triangles in $H$.  It is known that finding
the maximum number of vertex-disjoint triangles in a graph is
APX-hard~\cite{caprara02}. However, no explicit lower bound was known.


\begin{theorem}
	It is NP-hard to approximate the maximum number of vertex-disjoint triangles in a graph with a factor of $\left(\frac{362}{361}-\epsilon\right)$ for arbitrarily small constant $\epsilon>0$. The result holds even for graphs with maximum degree four.
\end{theorem}

\begin{figure}
	\centering
	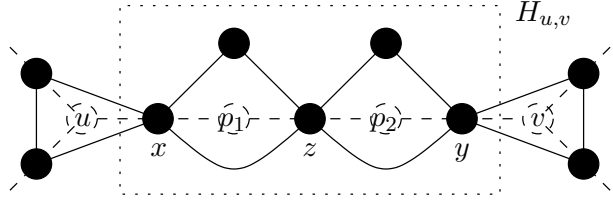
	\caption{Construction of graph $H$ (solid) from $G'$ (dashed) and definition of $H_{uv}$ (dotted).}
	\label{fig:edge_dual}
\end{figure}

Now, consider \OCP in $H$, \ie, finding the maximum number of vertex-disjoint odd cycles.
We prove that there is always an optimal solution consisting of triangles only. Assume the contrary, \ie, the optimal solution
includes a cycle $C$ of length at least $5$.
We distinguish three cases. \\
First, $C$ does not contain any vertex other than those of type $x$ and $y$ (see Figure~\ref{fig:edge_dual}). Then, $C$ must be a triangle.
Second, $C$ is fully contained in some $H_{uv}$. One easily checks that also in this case $C$ must be a triangle.
Third, $C$ is not contained in any $H_{uv}$.
Then $C$ has to pass through the node $z$ of some $H_{uv}$ and hence the solution does not include any triangle in such $H_{uv}$.
Substitute $C$ by any of the two triangles in $H_{uv}$ to obtain another optimal solution to \OCP.
Hence there is an optimal solution to \OCP in $H$ only containing triangles.
Therefore we have the following hardness result.

\begin{corollary}
	It is NP-hard to approximate \OCP
	with a factor of $\left(\frac{362}{361}-\epsilon\right)$ for arbitrarily small constant $\epsilon>0$. The result holds even for graphs with maximum degree four.
\end{corollary}

\subsection{From \OCP to \MVD and \MVS}


Consider the node-edge incidence matrix $A$ of $H$: this is a $1320n\times 2112n$ matrix.
From what was argued above, the maximum number of vertex-disjoint odd cycles of $361n$ (resp., $362n$)
translates into subdeterminants of size $2^{361n}$ (resp., $2^{362n}$). This allows us to prove a hardness of $2^n$ when the dimension is $d=1320n$.

\begin{theorem}
	It is NP-hard to approximate \MVD with a factor of $\left(2^{1/1320}-\epsilon\right)^d$ for arbitrarily small constant $\epsilon>0$.
	The hardness even holds when restricted to node-edge incidence matrices of graphs with maximum degree four.
\end{theorem}
\begin{proof}
	Due to the above construction of $H$, it is NP-hard to distinguish between the case when $\ocp(H)\le (1+\epsilon) 361n$ and $\ocp(H)\ge (1-\epsilon)362n$
	for arbitrarily small constant $\epsilon>0$. Hence, for \MVD on $A$ we obtain a gap of
	\begin{align*}
		2^{(1-\epsilon) 362n - (1+\epsilon) 361n} = 2^{n(1-723\epsilon)} = \left(2^{1-\epsilon'}\right)^{\frac{1}{1320}d} = \left(2^{1/1320}-\epsilon''\right)^d,
	\end{align*}
where $\epsilon'=723\epsilon$ and $\epsilon''=2^{1/1320} \left(1-2^{-\frac{\epsilon'}{1320}}\right)>0$.
\end{proof}

We now derive a similar inapproximability result for \MVS.

\begin{corollary}
	It is NP-hard to approximate \MVS~with a factor of $\left(2^{1/1320}-\epsilon\right)^d$ for arbitrarily small constant $\epsilon>0$.
\end{corollary}
\begin{proof}
We show that, if we were able to solve \MVS on $n$ vectors in $\Q^d$ with an approximation factor $\alpha(d)$, then we would be able to solve \MVD with input $A \in \Q^{d \times n}$ with an approximation factor of $\alpha(d) \cdot (d+1)$.\footnote{It is not difficult to prove that $d+1$ can be replaced by $d$. As this is not crucial for our proof, we leave the details to the interested reader.}
As we proved that \MVD is inapproximable up to a factor $c^d$ for some constant $c>1$ unless $P=NP$, we conclude an inapproximability for \MVD of $c^{d}/(d+1)=\Omega(\hat c^d)$ for any $\hat c<c$, unless $P=NP$.

Without loss of generality, let $a_1,\dots, a_d$ be an optimal solution to \MVD. Consider the \MVS instance with input $0,a_1,\dots,a_n$, and let $S$ be the simplex output by the algorithm. Note that $\conv\{0,a_1,\dots, a_d\}$ is a feasible solution to \MVS, hence $|\det (a_1,\dots,a_d)| / d! \leq \alpha(d) \cdot \vol(S)$. Now consider the triangulation of $\conv\{S,\{0\}\}$ into simplices $S_1,\dots, S_{d+1}$ containing the origin and $d$ of the vertices of $S$. We obtain
$$ \vol(S) \leq \vol(\conv\{S,\{0\}\}) = \sum_{i=1}^{d+1} \vol(S_i) \leq (d+1) \vol (S'),$$
where $S'$ is the simplex among $S_1,\dots, S_{d+1}$ of maximum volume. Note moreover that the submatrix $A'$ associated to the non-zero vertices of $S'$ is a feasible solution to the original \MVD problem. We deduce
$$\frac{|\det (a_1,\dots,a_d)|}{d!} \leq \alpha(d) \cdot \vol(S) \leq \alpha(d) \cdot (d+1) \cdot \vol(S') = \alpha(d) \cdot (d+1) \frac{|\det(A')|}{d!}.$$
Hence, we can output $A'$ and obtain the required approximation.
\end{proof}

\begin{remark} We remark that the construction of $G'$ and $H$ in
  Section \ref{sec:ss-to-ocp} can be improved. In fact, we do not need
  to subdivide \emph{every} edge of $G$ twice. It is sufficient to
  subdivide edges so that for every vertex of $G$, two of its incident
  edges are subdivided twice.  Hence, we can leave a maximum matching
  of $G$ untouched. Since we can compute a maximum matching in $G$ in
  polynomial time and every 3-regular graph of $\ell$ vertices has a
  matching of size $\frac{7}{16}\ell$~\cite{hobbs82}, we obtain the
  following slight improvements:

For every $\epsilon >0$, it is NP-Hard to approximate
\begin{itemize}
\item the maximum number of vertex-disjoint triangles in a graph and \OCP within a factor of $\left(\frac{285}{284}-\epsilon\right)$;
\item \MVD and \MVS with a factor of $\left(2^{1/1012}-\epsilon\right)^d$.
\end{itemize}
\end{remark}

\subsection{Tightness for instances with $n = O(d)$}

If the number of columns is linear in the number of rows, then a
better approximation result than that of Theorem~\ref{thr:2} is possible for \MVS and \MVD.
The next theorem is a
consequence of a recent result of Despande and
Rademacher~\cite{deshpande10}. These authors have shown how to
randomly sample from the set of bases of a given matrix $A$ such that the
probability of sampling a particular basis $B$ is proportional to
$\det(A_B^TA_B)$. In fact,  their algorithm is more general as it can
also handle $k$-subsets of linearly independent columns. Assume now
that $n \leq \alpha \cdot d$ with some $\alpha \in \N$. Then the
number of bases $|\B|$ is bounded by $\binom{\alpha \cdot d}{d} \leq (e \cdot
\alpha)^d$. We deduce
\begin{displaymath}
	\mathbb{E}(|\det(A_B)|) \cdot (e \alpha)^d \geq \Dm.
\end{displaymath}
The
claimed result then follows by repeated sampling and picking the
largest basis.

\begin{theorem}
  \label{thr:3}
  If $n = O(d)$, then there exists a randomized algorithm for \MVD
  with approximation ratio $\bar{c}^d$ for some constant $\bar{c}$
  that depends on the constant in the $O$-notation.
\end{theorem}

An alternative proof of this assertion relies on the fact that a
random point in the zonotope $\{y = A \cdot x \colon 0 \leq x \leq
1\}$ can be efficiently
sampled~\cite{JACM::DyerFK1991,MR2205290,lovasz1993random,lovasz2006hit}. It is folklore
that a zonotope can be partitioned into parallelepipeds generated by the
bases of $A$, \ie, each parallelipiped can be mapped one-to-one to a
basis. We can then identify the parallelepiped where the sampled point
resides. This results in sampling each basis with a probability
proportional to its determinant. We then obtain an estimation of
$\mathbb{E}(|\det(A_B)|)$, as required.


\small

\begin{thebibliography}{10}

\bibitem{bansal2010constructive}
Nikhil Bansal.
\newblock Constructive algorithms for discrepancy minimization.
\newblock In {\em Proceedings of the 51st annual IEEE Symposium on Foundations
  of Computer Science}, pages 3--10. IEEE, 2010.

\bibitem{berman03}
Piotr Berman and Marek Karpinski.
\newblock Improved approximation lower bounds on small occurrence optimization.
\newblock {\em Electronic Colloquium on Computational Complexity}, 10(008),
  2003.

\bibitem{brieden00}
Andreas Brieden, Peter Gritzmann, and Victor Klee.
\newblock Oracle-polynomial-time approximation of largest simplices in convex
  bodies.
\newblock {\em Discrete Mathematics}, 221(1–3):79--92, 2000.

\bibitem{caprara02}
Alberto Caprara and Romeo Rizzi.
\newblock Packing triangles in bounded degree graphs.
\newblock {\em Information Processing Letters}, 84(4):175--180, 2002.

\bibitem{civril09}
Ali \c{C}ivril and Malik Magdon-Ismail.
\newblock On selecting a maximum volume sub-matrix of a matrix and related
  problems.
\newblock {\em Theoretical Computer Science}, 410(47–49):4801--4811, 2009.

\bibitem{civril13}
Ali \c{C}ivril and Malik Magdon-Ismail.
\newblock Exponential inapproximability of selecting a maximum volume
  sub-matrix.
\newblock {\em Algorithmica}, 65:159--176, 2013.

\bibitem{chazelle2000discrepancy}
Bernard Chazelle.
\newblock {\em The discrepancy method: randomness and complexity}.
\newblock Cambridge University Press, 2000.

\bibitem{deshpande10}
Amit Deshpande and Luis Rademacher.
\newblock Efficient volume sampling for row/column subset selection.
\newblock In {\em Proceedings of the 51st annual IEEE Symposium on Foundations
  of Computer Science}, pages 329--338, 2010.

\bibitem{dudley74}
R.M. Dudley.
\newblock Metric entropy of some classes of sets with differentiable
  boundaries.
\newblock {\em Journal of Approximation Theory}, 10:227--236, 1974.

\bibitem{JACM::DyerFK1991}
Martin Dyer, Alan Frieze, and Ravi Kannan.
\newblock A random polynomial-time algorithm for approximating the volume of
  convex bodies.
\newblock {\em Journal of the ACM}, 38(1):1--17, 1991.

\bibitem{fiorini07}
Samuel Fiorini, Nadia Hardy, Bruce Reed, and Adrian Vetta.
\newblock Approximate min-max relations for odd cycles in planar graphs.
\newblock {\em Mathematical Programming}, 110(1, Ser. B):71--91, 2007.

\bibitem{salavatipour11}
Zachary Friggstad and Mohammad~R. Salavatipour.
\newblock Approximability of packing disjoint cycles.
\newblock {\em Algorithmica}, 60(2):395--400, 2011.

\bibitem{gritzmann94}
Peter Gritzmann and Victor Klee.
\newblock On the complexity of some basic problems in computational convexity.
  {I}. {C}ontainment problems.
\newblock {\em Discrete Mathematics}, 136(1-3):129--174, 1994.

\bibitem{gritzmann95}
Peter Gritzmann, Victor Klee, and David Larman.
\newblock Largest {$j$}-simplices in {$n$}-polytopes.
\newblock {\em Discrete \& Computational Geometry}, 13(3-4):477--515, 1995.

\bibitem{grossman95}
Jerrold~W. Grossman, Devadatta~M. Kulkarni, and Irwin~E. Schochetman.
\newblock On the minors of an incidence matrix and its {S}mith normal form.
\newblock {\em Linear Algebra and its Applications}, 218:213--224, 1995.

\bibitem{GroetschelLovaszSchrijver88}
Martin Gr{\"o}tschel, L\'aszl\'o Lov{\'a}sz, and Alexander Schrij\-ver.
\newblock {\em Geometric Algorithms and Combinatorial Optimization}, volume~2
  of {\em Algorithms and Combinatorics}.
\newblock Springer, 1988.

\bibitem{hardy05}
Nadia Hardy.
\newblock Odd cycles in planar graphs.
\newblock Master's thesis, Department of Mathematics and Statistics, McGill
  University, Montreal, Canada, 2005.

\bibitem{hastad97}
Johan H{\aa}stad.
\newblock Some optimal inapproximability results.
\newblock {\em Journal of the ACM}, 48(4):798--859, 2001.

\bibitem{hobbs82}
Arthur~M. Hobbs and Edward Schmeichel.
\newblock On the maximum number of independent edges in cubic graphs.
\newblock {\em Discrete Mathematics}, 42(2–3):317--320, 1982.

\bibitem{hochbaum1990convex}
Dorit~S. Hochbaum and J.~George Shanthikumar.
\newblock Convex separable optimization is not much harder than linear
  optimization.
\newblock {\em Journal of the ACM}, 37(4):843--862, 1990.

\bibitem{KarmakarKarp82}
Narendra Karmarkar and Richard~M. Karp.
\newblock An efficient approximation scheme for the one-dimensional binpacking
  problem.
\newblock In {\em Proceedings of the 23rd annual Symposium on Foundations of
  Computer Science}, pages 312--320, 1982.

\bibitem{KR10}
{Ken-ichi} Kawarabayashi and Bruce Reed.
\newblock Odd cycle packing.
\newblock In {\em Proceedings of the 42nd annual ACM Symposium on Theory of
  Computing}, pages 695--704. ACM, 2010.

\bibitem{khachiyan95}
Leonid~G. Khachiyan.
\newblock On the complexity of approximating extremal determinants in matrices.
\newblock {\em Journal of Complexity}, 11(1):138--153, 1995.

\bibitem{khachiyan96}
Leonid~G. Khachiyan.
\newblock Rounding of polytopes in the real number model of computation.
\newblock {\em Mathematics of Operations Research}, 21(2):307--320, 1996.

\bibitem{khachiyan93}
Leonid~G. Khachiyan and Michael~J. Todd.
\newblock On the complexity of approximating the maximal inscribed ellipsoid
  for a polytope.
\newblock {\em Mathematical Programming}, 61(2, Ser. A):137--159, 1993.

\bibitem{koutis06}
Ioannis Koutis.
\newblock Parameterized complexity and improved inapproximability for computing
  the largest j-simplex in a {$V$}-polytope.
\newblock {\em Information Processing Letters}, 100(1):8--13, 2006.

\bibitem{kral12}
Daniel Kr\'al', Jean-S\'ebastien Sereni, and Ladislav Stacho.
\newblock Min-max relations for odd cycles in planar graphs.
\newblock {\em SIAM Journal on Discrete Mathematics}, 26(3):884--895, 2012.

\bibitem{lovasz1993random}
L\'aszl\'o Lov{\'a}sz and Mikl\'os Simonovits.
\newblock Random walks in a convex body and an improved volume algorithm.
\newblock {\em Random structures \& algorithms}, 4(4):359--412, 1993.

\bibitem{lovasz1986discrepancy}
L\'aszl\'o Lov{\'a}sz, Joel Spencer, and Katalin Vesztergombi.
\newblock Discrepancy of set-systems and matrices.
\newblock {\em European Journal of Combinatorics}, 7(2):151--160, 1986.

\bibitem{lovasz2006hit}
L{\'a}szl{\'o} Lov{\'a}sz and Santosh Vempala.
\newblock Hit-and-run from a corner.
\newblock {\em SIAM Journal on Computing}, 35(4):985--1005, 2006.

\bibitem{MR2205290}
L{\'a}szl{\'o} Lov{\'a}sz and Santosh Vempala.
\newblock Simulated annealing in convex bodies and an {$O^*(n^4)$} volume
  algorithm.
\newblock {\em Journal of Computer and System Sciences}, 72(2):392--417, 2006.

\bibitem{matousek02}
Ji{\v{r}}{\'{\i}} Matou{\v{s}}ek.
\newblock {\em Lectures on discrete geometry}, volume 212 of {\em Graduate
  Texts in Mathematics}.
\newblock Springer-Verlag, New York, 2002.

\bibitem{matouvsek2013determinant}
Ji{\v{r}}{\'\i} Matou{\v{s}}ek.
\newblock The determinant bound for discrepancy is almost tight.
\newblock {\em Proceedings of the American Mathematical Society},
  141(2):451--460, 2013.

\bibitem{nikolov2013approximating}
Aleksandar Nikolov and Kunal Talwar.
\newblock Approximating hereditary discrepancy via small width ellipsoids.
\newblock {\em arXiv preprint arXiv:1311.6204}, 2013.

\bibitem{nikolov2013geometry}
Aleksandar Nikolov, Kunal Talwar, and Li~Zhang.
\newblock The geometry of differential privacy: the sparse and approximate
  cases.
\newblock In {\em Proceedings of the 45th annual ACM Symposium on Theory of
  Computing}, pages 351--360. ACM, 2013.

\bibitem{packer02}
Asa Packer.
\newblock {NP}-hardness of largest contained and smallest containing simplices
  for {$V$}- and {$H$}-polytopes.
\newblock {\em Discrete \& Computational Geometry}, 28(3):349--377, 2002.

\bibitem{packer04}
Asa Packer.
\newblock Polynomial-time approximation of largest simplices in
  {$V$}-polytopes.
\newblock {\em Discrete Applied Mathematics}, 134(1-3):213--237, 2004.

\bibitem{papadimitriou84}
Christos~H. Papadimitriou.
\newblock The largest subdeterminant of a matrix.
\newblock {\em Bulletin of the Greek Mathematical Society}, 25:95--105, 1984.

\bibitem{rothvoss2013approximating}
Thomas Rothvo{\ss}.
\newblock Approximating bin packing within {O(log OPT* log log OPT)} bins.
\newblock In {\em Proceedings of the 54th annual IEEE Symposium on Foundations
  of Computer Science}, pages 20--29. IEEE, 2013.

\bibitem{Schrijver98}
Alexander Schrijver.
\newblock {\em Theory of Linear and Integer Programming}.
\newblock Wiley-Interscience Series In Discrete Mathematics And Optimization,
  1998.

\bibitem{MR579077}
Paul~D. Seymour.
\newblock Decomposition of regular matroids.
\newblock {\em Journal of Combinatorial Theory. Series B}, 28(3):305--359,
  1980.

\bibitem{spencer1985six}
Joel Spencer.
\newblock Six standard deviations suffice.
\newblock {\em Transactions of the American Mathematical Society},
  289(2):679--706, 1985.

\end{thebibliography}
\end{document}